\documentclass[a4paper,11pt,fleqn]{article}

\usepackage{amsmath,amssymb,amsthm,enumerate,cite,graphicx} 

\newcommand{\noprint}[1]{}
\newcommand{\const}{\mathop{\rm const}\nolimits}

\newcommand{\thetbn}{\arabic{nomer}}

\newcommand{\todo}[1][\null]{\ensuremath{\clubsuit}}

\newcounter{tbn}

\newcounter{mcasenum}

\newtheorem{theorem}{Theorem}

\newtheorem{corollary}{Corollary}

\newtheorem*{proposition*}{Proposition}
{\theoremstyle{definition}

}

\begin{document}
\begin{center}
\par\noindent {\LARGE\bf
Numerical solutions of boundary value problems\\ for variable coefficient generalized KdV  equations\\  using Lie symmetries
\par}

{\vspace{4mm}\par\noindent\large  O.\,O.~Vaneeva$^{\dag 1}$, N.C. Papanicolaou$^{\ddag 2}$, M.A. Christou$^{\ddag 3}$ and C. Sophocleous$^{\S 4}$
\par\par}
\end{center}
{\par\noindent\it
${}^\dag$\ Institute of Mathematics of the National Academy of Sciences of Ukraine,\\[0.75ex]
\phantom{${}^\dag$} 3 Tereshchenkivska Str., Kyiv-4, 01601, Ukraine\\[1ex]
${}^\ddag$\ Department of Mathematics, University of Nicosia, P.O. Box 24005, 1700 Nicosia, Cyprus\\[1ex]
${}^\S$\ Department of Mathematics and Statistics, University of Cyprus, Nicosia CY 1678, Cyprus
}
\begin{center}
{\par\noindent\small
$^1$vaneeva@imath.kiev.ua, $^2$papanicolaou.n@unic.ac.cy, $^3$christou.ma@unic.ac.cy,
$^4$christod@ucy.ac.cy
\par}
\end{center}

{\vspace{3mm}\par\noindent\hspace*{8mm}\parbox{146mm}{\small The exhaustive group classification of a class of variable coefficient generalized KdV equations is presented, which completes and enhances results existing in the literature. Lie symmetries are used for solving an initial and boundary value problem for certain subclasses of the above class. Namely, the found Lie symmetries are applied in order to reduce the initial and boundary value problem for the generalized KdV equations (which are PDEs) to an initial value problem for  nonlinear third-order ODEs. The latter problem is solved numerically using the finite difference method. Numerical solutions are computed and the vast parameter space is studied.
}\par\vspace{2mm}}
\section{Introduction}

Nonlinear partial differential equations (PDEs) are often used in the physical sciences and enginee\-ring to  model real-world phenomena. Unfortunately, there is still no general theory of solving such equations. One of the most efficient techniques for construction of solutions is the Lie reduction method, based on the usage of Lie symmetries that correspond to Lie groups of continuous point transformations~\cite{Olver1986,Ovsiannikov1982}.
The Lie reduction method is algorithmic and powerful,
in particular, any ($1{+}1$)-dimensional PDE admitting a one-parameter Lie symmetry group (acting regularly and transversally on a manifold defined by the PDE~\cite{Olver1986}) can always be reduced to an ordinary differential equation (ODE).

 In applications,  one is usually interested   in a solution of a given PDE satisfying some
initial condition  or/and a boundary condition.
In a recent paper~\cite{VSL}, Lie symmetries were successfully  applied to solve an initial  and boundary value problem (IBVP) for a generalized Burgers equation arising in nonlinear acoustics. Namely, the IBVP for a generalized Burgers equation was reduced to an initial value problem (IVP) for a related nonlinear second-order ODE. As a result, a closed-form solution of the IBVP for the generalized Burgers equation was found.
Motivated by that work, we intend to apply  Lie symmetries  to construct solutions for IBVP for the variable coefficient generalized  Korteweg--de Vries (KdV) equation,
   \begin{equation}\label{EqSenthilkumaran&Pandiaraja&Vaganan2008}
      u_t+u^nu_x+h(t)u+g(t)u_{xxx}=0,\qquad ng\neq0,
   \end{equation}
arising in several applications (see~\cite{Senthilkumaran&Pandiaraja&Vaganan2008} and references therein). For this purpose we carry out an exhaustive group classification of equations from this class. In other words, we at first find a Lie invariance algebra admitted by any equation in the class, the so-called kernel algebra, and then classify all possible cases of extension of Lie invariance algebras of such equations with respect to
the equivalence group of the class~\cite{Ovsiannikov1982}.
 Some cases of Lie symmetry extension for~\eqref{EqSenthilkumaran&Pandiaraja&Vaganan2008} were found in~\cite{Senthilkumaran&Pandiaraja&Vaganan2008}, namely,  the cases $h={\rm const}$ and $h=1/(at+b)$ with $a$ and $b$ being constants.
 Here we present a complete group classification taking advantage of the use of equivalence transformations (this opportunity was neglected in~\cite{Senthilkumaran&Pandiaraja&Vaganan2008}). We point out that complete group classifications of class~\eqref{EqSenthilkumaran&Pandiaraja&Vaganan2008} for $n=1$ and $n=2$ were obtained in~\cite{Popovych&Vaneeva2010,Vaneeva2012} (see also~\cite{Vaneeva2013}). The results were presented there in two ways: with respect to corresponding equivalence groups and
 using no equivalence. We would like to mention that in~\cite{Gungor&Lahno&Zhdanov2004,Magadeev1993} group classifications for more general classes that include
 class~\eqref{EqSenthilkumaran&Pandiaraja&Vaganan2008} were carried out. However it seems to be inconvenient
  to derive group classifications for class~\eqref{EqSenthilkumaran&Pandiaraja&Vaganan2008} using those results obtained up to a very wide equivalence group.

 Note that the more general class of the form
\begin{equation}\label{EqSenthilkumaran&Pandiaraja&Vaganan2008fgh}
u_t+f(t)u^nu_x+h(t)u+g(t)u_{xxx}=0,\quad n, f, g \neq0,
\end{equation}
reduces to class~\eqref{EqSenthilkumaran&Pandiaraja&Vaganan2008} via the change of the variable $t$.
That is why without loss of generality it is sufficient to study  class~\eqref{EqSenthilkumaran&Pandiaraja&Vaganan2008},
since all results on exact solutions, symmetries, conservation laws, etc.
for class~\eqref{EqSenthilkumaran&Pandiaraja&Vaganan2008fgh} can be derived form those
obtained for~\eqref{EqSenthilkumaran&Pandiaraja&Vaganan2008}.

 After group classification is carried out, the derived Lie symmetries are employed  to construct group-invariant (similarity) solutions for equations \eqref{EqSenthilkumaran&Pandiaraja&Vaganan2008}. This is achieved by constructing corresponding transformations (such a transformation is usually called an ansatz) that reduce equations of the form ~\eqref{EqSenthilkumaran&Pandiaraja&Vaganan2008} which admit Lie symmetries into ODEs. The complete list of such transformations is achieved by constructing the optimal system of subalgebras of the maximal Lie invariance algebras~\cite{Olver1986}.
 The same transformations are used to reduce an IBVP for a generalized KdV equation into an IVP for a corresponding nonlinear ODE.
 Here we require that both equation and additional conditions are invariant under the Lie group of infinitesimal transformations. This is the ``classical'' technique
 described in~\cite{Bluman&Cole1969,Bluman&Anco2002,Bluman1974}. A discussion on some other approaches can be found in~\cite{AbdelMalek1998,VSL}.
 The reduced IVP for ODE is solved numerically using a finite difference algorithm with fixed-point iterations. We present solutions obtained using specified values of the governing parameters. Then the solution of the original IBVP is recovered with the aid of the transformations in question. Therefore, we show that Lie symmetry methods are useful even in those cases when reduced ODEs cannot be solved analytically but only numerically.

\section{Equivalence transformations}
In order to solve a group classification problem it is important to derive the
transformations that preserve the differential structure of
the class of differential equations under consideration~\cite{Ovsiannikov1982}. Such transformations form a group and are called equivalence transformations.
The knowledge of transformations from the  equivalence group often gives an opportunity to essentially simplify a
group classification problem and to present final results in a concise form.
Moreover, sometimes this appears to be a crucial point in
the exhaustive solution of such problems (see, e.g.,~\cite{vane2012b,Vaneeva2012} and references therein).
We find equivalence transformations in class~\eqref{EqSenthilkumaran&Pandiaraja&Vaganan2008} using the direct method~\cite{Kingston&Sophocleous1998}.
\begin{theorem} The equivalence group~$G^{\sim}$ of class~\eqref{EqSenthilkumaran&Pandiaraja&Vaganan2008} consists of the transformations
\begin{gather}\label{theorem1}
\begin{array}{l}
\tilde t=\varepsilon_1 \int \theta^{-n} dt+\varepsilon_2,\quad \tilde x=\varepsilon_1x+\varepsilon_0,\quad
\tilde u=\theta u, \\[1ex]
\tilde g=\varepsilon_1^{\,2}\,\theta^n\, g,\quad \tilde h=\dfrac{\theta^n}{\varepsilon_1}\left(h-\dfrac{\theta_t}{\theta}\right),\quad \tilde n=n,
\end{array}
\end{gather}
where  $\varepsilon_j,$ $j=0,1,2,$ are arbitrary constants with
$\varepsilon_1\not=0$; $\theta=\theta(t)$ is an arbitrary nonvanishing smooth function.
\end{theorem}
\begin{proof}
Suppose
that an equation of the form~\eqref{EqSenthilkumaran&Pandiaraja&Vaganan2008} is connected with an equation
\begin{gather}
\label{GBE_tilde}
{\tilde u}_{\tilde t}+\tilde u^{\tilde n}{\tilde u}_{\tilde x}+\tilde h(\tilde t)\tilde u+
\tilde g(\tilde t){\tilde u}_{\tilde x\tilde x\tilde x}=0
\end{gather}
from the same class by a~point transformation
$\tilde t=T(t,x,u)$,
$\tilde x=X(t,x,u)$,
$\tilde u=U(t,x,u)$,
where $|\partial(T,X,U)/\partial(t,x,u)|\ne0$.
It is proved in~\cite{VPS2013} that admissible transformations of evolution equations of the form $u_t=F(t)u_n+G(t,x,u,u_1,\dots,u_{n-1}),$ $F\ne0,$ $ G_{u_iu_{n-1}}=0,$ $i=1,\dots,n-1,$
where $n\geqslant 2$, obey the restrictions $T_x=T_u=X_u=X_{xx}=U_{uu}=0$.
Therefore, for the class~\eqref{EqSenthilkumaran&Pandiaraja&Vaganan2008} it is sufficient to consider transformations of the form
\begin{gather*}
\tilde t=T(t),\quad \tilde x=X^1(t)x+X^0(t),\quad \tilde u=U^1(t,x)u+U^0(t,x),
\end{gather*}
where $T_tX^1U^1\neq0$.
Expressing $\tilde u$ and its derivatives ${\tilde u}_{\tilde t}$, ${\tilde u}_{\tilde x}$, and ${\tilde u}_{\tilde x\tilde x\tilde x}$ in~\eqref{GBE_tilde} in terms of the untilded variables, we get the equation
\begin{gather*}
X^1(U^1_tu+U^1u_t+U^0_t)-(X^1_tx+X^0_t)(U^1_xu+U^1u_x+U^0_x)+T_t(U^1u+U^0)^{\tilde n}(U^1_xu+U^1u_x\\+U^0_x)+\tilde h T_tX^1(U^1u+U^0)+\tilde g \frac{T_t}{(X^1)^2}(U^1_{xxx}u+3U^1_{xx}u_x+3U^1_{x}u_{xx}+U^1u_{xxx}+U^0_{xxx})=0.
\end{gather*}
We then substitute $u_t=-u^nu_x-h(t)u-g(t)u_{xxx}$ into the latter equation
in order to confine it to the manifold defined by~\eqref{EqSenthilkumaran&Pandiaraja&Vaganan2008} in the third-order jet space
with the independent variables $(t,x)$ and the dependent variable~$u$. Splitting the obtained identity with respect to the derivatives $u_x$, $u_{xx}$ and $u_{xxx}$, we get the following conditions $\tilde n=n$, $U^1_x=0$, $\tilde g=(X^1)^3T_t^{-1}g$, $T_t=X^1(U^1)^{-n}$; $U^0=0$ for all $n$ except $n=1$. As the group classification for class~\eqref{EqSenthilkumaran&Pandiaraja&Vaganan2008} with $n=1$ is known, we assume that $n\neq1$.\footnote{Note that if $n=1$, then class~\eqref{EqSenthilkumaran&Pandiaraja&Vaganan2008} admits a wider equivalence group (see~\cite{Vaneeva2013}).}
Then, the rest of the determining equations are
$X^1_t=X^0_t=0$, and $\tilde hT_tU^1+U^1_t-hU^1=0.$ Solving these equations we obtain $X^i=\varepsilon_i$, $i=0,1$, where $\varepsilon_1$ and $\varepsilon_0$  are arbitrary constants with $\varepsilon_1\neq0;$ $U^1=\theta(t)$ is an arbitrary nonvanishing smooth function of~$t$. Then, $T=\varepsilon_1 \int \theta^{-n} dt+\varepsilon_2$,
where $\varepsilon_2$ is an arbitrary constant. Here and in what follows an integral with respect to~$t$ should be interpreted as a fixed antiderivative. The transformation components for $g$ and $h$ take the form presented in~\eqref{theorem1}. Theorem~1 is  proved.
\end{proof}

Now we can use equivalence transformations~\eqref{theorem1} to gauge one of the arbitrary elements $g$ or $h$ to a simple constant value.
It was shown in~\cite{Popovych&Vaneeva2010} that the parameter-function~$h$ in~\eqref{EqSenthilkumaran&Pandiaraja&Vaganan2008} can be set equal to zero by the point transformation
\begin{equation}\label{EqSenthilkumaran&Pandiaraja&Vaganan2008Trans}
\tilde t=\int e^{-n\!\int h(t)\,dt}\,dt,\quad
\tilde x=x,\quad
\tilde u=e^{\,\int\! h(t)\,dt} u,
\end{equation}
and the transformed value of the arbitrary element~$g$ is $\tilde g(\tilde t)=e^{\,n\!\int h(t)\,dt}g(t)$. This transformation can be easily found from Theorem~2 setting $\tilde h=0$ in~\eqref{theorem1} and solving the obtained equation for~$\theta$.
The fact that the arbitrary element $h$ can always be set to zero means that fixing the arbitrary element~$h$ cannot lead to cases of equations~\eqref{EqSenthilkumaran&Pandiaraja&Vaganan2008}
with special symmetry properties. So, without loss of generality we can restrict our investigation to the class
\begin{equation}\label{EqSenthilkumaran&Pandiaraja&Vaganan2008_h=0}
u_t+u^nu_x+g(t)u_{xxx}=0, \quad n, g \neq0.
\end{equation}
All results on Lie symmetries and exact solutions for equations~\eqref{EqSenthilkumaran&Pandiaraja&Vaganan2008} can be derived from similar results obtained for equations~\eqref{EqSenthilkumaran&Pandiaraja&Vaganan2008_h=0}.

There are no other point transformations except~\eqref{theorem1} that link equations from class~\eqref{EqSenthilkumaran&Pandiaraja&Vaganan2008} with $n\neq0,1$, therefore, this class is normalized (see related notions in~\cite{popo2010a,VPS2013}) with respect to its equivalence group $G^{\sim}$. It means that an equivalence group for~\eqref{EqSenthilkumaran&Pandiaraja&Vaganan2008_h=0} can be derived from Theorem~1 by simply setting $\tilde h=h=0$ (for classes which are not normalized this could lead to an incomplete result). Then we get that $\theta=\varepsilon_3$ is an arbitrary nonzero constant and the following statement is true.

\begin{corollary} The equivalence group~$G^{\sim}_0$ of class~\eqref{EqSenthilkumaran&Pandiaraja&Vaganan2008_h=0} is formed by the transformations
\[
\begin{array}{l}
\tilde t={\varepsilon_1}{\varepsilon_3^{\,-n}} t+\varepsilon_0,\quad \tilde x=\varepsilon_1x+\varepsilon_2,\quad
\tilde u=\varepsilon_3u, \\[1ex]
\tilde g=\varepsilon_1^{\,2}\varepsilon_3^{\,n}g,\quad \tilde n=n,
\end{array}
\]
where  $\varepsilon_j,$ $j=0,\dots,3,$ are arbitrary constants with
$\varepsilon_1\varepsilon_3\not=0$.
\end{corollary}

\section{Lie symmetries}
The group classification of class~\eqref{EqSenthilkumaran&Pandiaraja&Vaganan2008_h=0} up to $G^{\sim}_0$-equivalence coincides with
the group classification of class~\eqref{EqSenthilkumaran&Pandiaraja&Vaganan2008} up to $G^{\sim}$-equivalence.
We carry out the group classification of class~\eqref{EqSenthilkumaran&Pandiaraja&Vaganan2008_h=0}
using the  classical algorithm~\cite{Olver1986}.
Namely, we
search for Lie symmetry generators of the form $\Gamma=\tau(t,x,u)\partial_t+\xi(t,x,u)\partial_x+\eta(t,x,u)\partial_u$
and require that
\begin{equation}\label{c2}
\Gamma^{(3)}\{u_t+u^nu_x+g(t)u_{xxx}\}=0
\end{equation}
identically, modulo equation~\eqref{EqSenthilkumaran&Pandiaraja&Vaganan2008_h=0}. Here  $\Gamma^{(3)}$ is the third prolongation of the operator~$\Gamma$~\cite{Olver1986,Ovsiannikov1982},
$\Gamma^{(3)}=\Gamma+\eta^t\partial_{u_t}+\eta^x\partial_{u_x}+\eta^{xxx}\partial_{u_{xxx}}$, where
\begin{gather*}
\eta^t=D_t(\eta)-u_tD_t(\tau)-u_xD_t(\xi),\quad
\eta^x=D_x(\eta)-u_tD_x(\tau)-u_xD_x(\xi),\\
\eta^{xx}=D_x(\eta^x)-u_{tx}D_x(\tau)-u_{xx}D_x(\xi),\quad
\eta^{xxx}=D_x(\eta^{xx})-u_{txx}D_x(\tau)-u_{xxx}D_x(\xi),
\end{gather*}
$D_t=\partial_t+u_t\partial_{u}+u_{tt}\partial_{u_t}+u_{tx}\partial_{u_x}+\dots{}$ and
$D_x=\partial_x+u_x\partial_{u}+u_{tx}\partial_{u_t}+u_{xx}\partial_{u_x}+\dots{}$
are operators of the total differentiation with respect to~$t$ and~$x$, respectively.
The infinitesimal invariance criterion~\eqref{c2} implies that
\[
\tau=\tau(t),\quad
\xi=\xi(t,x), \quad
\eta=\eta^1(t,x)u+\eta^0(t,x),
\]
where $\tau$, $\xi$, $\eta^1$ and $\eta^0$ are arbitrary smooth functions of their variables.
The rest of the determining equations have the form
\begin{gather}\label{deteq1a}
3g\xi_x=(g\tau)_t,\qquad \eta^1_x=\xi_{xx},\\\label{deteq2a}
\eta^1_xu^{n+1}+\eta^0_xu^{n}+(\eta^1_t+g\eta^1_{xxx})u+\eta^0_t+g\eta^0_{xxx}=0,\\\label{deteq3a}
(\tau_t-\xi_x+n\eta^1)u^n+n\eta^0u^{n-1}-g\xi_{xxx}+3g\eta^1_{xx}-\xi_t=0.
\end{gather}
It is easy to see from~\eqref{deteq1a} that $\xi_{xx}=0$ and therefore $\eta^1=\eta^1(t).$
The second and third equations can be split with respect to different powers of $u$. Special cases of splitting arise if $n=0,1$.
If $n=0$ equations~\eqref{EqSenthilkumaran&Pandiaraja&Vaganan2008_h=0} are linear and are excluded from our current investigation. We also do not consider the case $n=1$ since the exhaustive group classification for this case is carried out in~\cite{Popovych&Vaneeva2010,Vaneeva2013}.  Therefore we only investigate cases for which $n\neq0,1$.

If $n\neq0,1,$ then~\eqref{deteq1a}--\eqref{deteq3a} imply $\eta^1=c_0,$ where $c_0$ is an arbitrary constant and $\eta^0=0$. Then
\[\tau=c_1t+c_2,\quad \xi=(c_1+nc_0)x+c_3,\quad\eta=c_0u,\]
where $c_i$, $i=0,\dots,3$, are arbitrary constants.
The classifying equation for $g$ has the form
\begin{equation}\label{c11}
(c_1t+c_2)g_t=(2c_1+3nc_0)g.
\end{equation}

Further analysis is carried out using  the method of furcate split~\cite{Popovych&Ivanova2004}.
For each operator~$\Gamma$ from maximal Lie invariance algebra $A^{\max}$ equation~\eqref{c11}
gives some equation on~$g$ of the general form
$
(a\,t+b)g_t=c\,g,
$
where $a,$ $b$ and  $c$ are constants. In general, for all operators from $A^{\max}$
the number $k$ of such independent equations is no greater than 2, otherwise
they form an incompatible system for $g$. Therefore, there exist three
distinct  cases for the value of $k$: $k=0,$ $k=1$ and $k=2$.

If $k=0,$ then~\eqref{c11} is not an equation but an identity and $c_0=c_1=c_2=0$.
So,  the kernel of maximal Lie invariance algebras of equations from~\eqref{EqSenthilkumaran&Pandiaraja&Vaganan2008_h=0} (i.e., the Lie invariance algebra admitted by~\eqref{EqSenthilkumaran&Pandiaraja&Vaganan2008_h=0} for arbitrary $g$) is the
one-dimensional algebra $\langle\partial_x\rangle$.

If~$k=1,$ then up to $G^{\sim}_0$-equivalence, $g$ is either a power function, $g=\varepsilon t^{\rho}$, or  exponential, $g=\varepsilon e^t$,  Here, $\rho$ and $\varepsilon$
are nonzero constants and $\varepsilon=\pm1\bmod G^{\sim}_0$.
If $g=\varepsilon  t^\rho$ and $\rho\neq0,$ then
\[\Gamma=c_1t\partial_t+\left(\tfrac{\rho+1}{3}c_1x+c_3\right)\partial_x+ \tfrac{\rho-2}{3n}c_1u\partial_u.\]
In the exponential case, $g=\varepsilon e^t$, the general form of the infinitesimal operator is
\[\Gamma=3nc_0\partial_t+(nc_0x+c_3)\partial_x+c_0u\partial_u.\]

In both these cases, the maximal Lie-invariance algebras are two-dimensional with basis operators presented in Cases 1 and 2 of Table~1. In Case 1 with $\rho=-1$ the Lie algebra is Abelian. The algebras adduced in Case 1 with $\rho\neq-1$ and Case 2 are non-Abelian.

If $k=2$, $g=\varepsilon=\const,$ where $\varepsilon=\pm1\bmod G^\sim_0$. The infinitesimal operator takes the form
\[\Gamma=(c_1t+c_2)\partial_t+\left(\tfrac13c_1x+c_3\right)\partial_x- \tfrac{2}{3n}c_1u\partial_u.\]
So, if $g$ is a nonzero constant, the maximal Lie invariance algebra of~\eqref{EqSenthilkumaran&Pandiaraja&Vaganan2008_h=0} with $n\neq1$ is three-dimensional
spanned by the operators presented in Case 3 of Table~1. This is the Lie algebra of the type $A^a_{3.5}$ with $a=1/3$~\cite{pate1977a}.
We summarize the above consideration by the statement.
\begin{theorem}
The kernel of the maximal Lie invariance algebras $A^{\rm max}$ of equations from class~\eqref{EqSenthilkumaran&Pandiaraja&Vaganan2008_h=0} {\rm(}resp.~\eqref{EqSenthilkumaran&Pandiaraja&Vaganan2008}{\rm)}
coincides with the one-dimensional algebra $\langle\partial_x\rangle$.
All possible $G_0^\sim$-inequivalent {\rm(}resp. $G^\sim$-inequivalent{\rm)} cases of extension of $A^{\rm max}$ are exhausted
by cases 1--3 of Table~\ref{TableLieSymG}.
\end{theorem}
\begin{table}[h!]\small \renewcommand{\arraystretch}{1.75}
\begin{center}
\setcounter{tbn}{-1}
\refstepcounter{table}\label{TableLieSymG}
\textbf{Table~\thetable.}
The group classification of class~\eqref{EqSenthilkumaran&Pandiaraja&Vaganan2008} up to $G^\sim$-equivalence.
\\[2ex]
\begin{tabular}{|c|c|l|}
\hline
no.&$g(t)$&\hfil Basis of $A^{\max}$ \\
\hline
\refstepcounter{tbn}\label{TableLieSym_ker}\thetbn&$\forall$&$\partial_x$\\
\hline
\refstepcounter{tbn}\label{TableLieSym_2op}\thetbn&
$\varepsilon t^\rho$&$\partial_x,\,3nt\partial_t+(\rho+1)n x\partial_x+(\rho-2) u\partial_u$\\
\hline
\refstepcounter{tbn}\label{TableLieSym_3op}\thetbn&$\varepsilon e^{t}$&
$\partial_x,\,3n\partial_t+nx\partial_x+u\partial_u$\\
\hline
\refstepcounter{tbn}\label{TableLieSymHF_const}\thetbn&$\varepsilon $&
$\partial_x,\,\partial_t,\,3nt\partial_t+nx\partial_x-2u\partial_u$\\
\hline
\end{tabular}
\\[2ex]
\parbox{150mm}{Here $h=0\bmod\, G^\sim$, $\varepsilon=\pm1\bmod\, G^\sim,$ $\rho$ is an arbitrary nonzero constant.}
\end{center}
\end{table}

In Table~2, we also adduce the results of the group classification of~\eqref{EqSenthilkumaran&Pandiaraja&Vaganan2008} without gauging of $h$ and~$g$ by equivalence transformations. The extended classification list can be derived from that presented in Table~1 using equivalence transformations (the detailed procedure is described in~\cite{Vaneeva2012,Vaneeva2013}). It is easy to see that Table 2 includes all cases presented in~\cite{Senthilkumaran&Pandiaraja&Vaganan2008} as special cases.
\begin{table}[h!]\small \renewcommand{\arraystretch}{2}
\begin{center}
\setcounter{tbn}{-1}
\refstepcounter{table}\label{TableLieSymHF2}
\textbf{Table~\thetable.}
The group classification of class~\eqref{EqSenthilkumaran&Pandiaraja&Vaganan2008} using no equivalence.
\\[2ex]
\begin{tabular}{|c|c|l|}
\hline
no.&$g(t)$&\hfil Basis of $A^{\max}$ \\
\hline
\refstepcounter{tbn}\label{TableLieSym_ker2}\thetbn&$\forall$&$\partial_x$\\
\hline
\refstepcounter{tbn}\label{TableLieSym_2op2}\thetbn&
$\lambda\big(\int e^{-n\!\int\! h(t)\, dt}dt+\kappa\big)^{\rho}\!\!e^{-n\!\int\! h(t)\! dt}$&$\partial_x,\,H\partial_t+n(\rho+1)x\partial_x+(\rho-2-h(t)H) u\partial_u$\\
\hline
\refstepcounter{tbn}\label{TableLieSym_3op2}\thetbn&$\lambda\:  e^{\int\left( m e^{-n\!\int\! h(t)\, dt}-n h(t)\right) dt}$&
$\partial_x,\,3ne^{n\!\int\! h(t) dt}\partial_t+mn x\partial_x+\big(m-3nh(t)e^{n\!\int\! h(t) dt}\big)u\partial_u$\\
\hline
\refstepcounter{tbn}\label{TableLieSymHF_const2}\thetbn&$\lambda\:  e^{-n\!\int\! h(t) dt}$&
$\partial_x,\,e^{n\!\int\! h(t) dt}\left(\partial_t-h(t)u\partial_u\right),\,H\partial_t+nx\partial_x-(2+h(t)H) u\partial_u$\\
\hline
\end{tabular}
\\[2ex]
\parbox{150mm}{Here $\lambda $, $\kappa$, $\rho$, and $m$ are arbitrary constants with $\lambda \rho m\neq0$. The function $h(t)$ is arbitrary in all cases and
 $H=3ne^{n\!\int\! h dt}\big(\int e^{-n\!\int\! h dt} dt+\kappa\big)$.
In Case~3, $\kappa=0$ in the formula for $H$.}
\end{center}
\end{table}

\section{Similarity solutions of the generalized KdV equations}
Lie symmetries provide us with an algorithmic technique for finding exact solutions using the reduction method~\cite{Olver1986,Ovsiannikov1982}. It was found by Lie that if one Lie symmetry generator of an ODE is known, than the order of this ODE can be reduced by one, and if we know a Lie symmetry generator for a $n$-dimensional PDE, then it can be reduced to a $n{-}1$-dimensional PDE. This is true for Lie symmetry generators corresponding to
one-parameter Lie symmetry groups that act regularly and transversally on a manifold defined by this PDE~\cite{Olver1986}.
So, in our case of $(1+1)$-dimensional PDEs it is enough to perform   reductions with respect to one-dimensional subalgebras of the found maximal Lie invariance algebras to get reductions to ODEs.

Reductions should be performed using subalgebras from the optimal system~\cite{Olver1986}. Optimal systems of one-dimensional subalgebras of Lie invariance algebras from Table~1 are presented in Table~3.

\begin{table}[h!]\small \renewcommand{\arraystretch}{1.65}
\begin{center}
\textbf{Table 3.} Optimal systems of subalgebras of $A^{\rm max}$  presented in Table 1.
\\[2ex]
\begin{tabular}
{|c|l|}
\hline
\hfil no.
&
\hfil Optimal system
\\
\hline
$1_{\rho\neq-1}$
&
${\mathfrak g}^{\,}_1=\langle\partial_x\rangle,
\quad
{\mathfrak g}^{\,}_{1.1}=\langle 3nt\partial_t+(\rho+1)n x\partial_x+(\rho-2) u\partial_u\rangle$
\\
\hline
$1_{\rho=-1}$
&
${\mathfrak g}^{\,}_1=\langle\partial_x\rangle,
\quad
{\mathfrak g}^a_{1.2}=\langle nt\partial_t+a\partial_x-u\partial_u\rangle$
\\
\hline
2
&
${\mathfrak g}^{\,}_1=\langle\partial_x\rangle,
\quad
{\mathfrak g}^{\,}_2=\langle3n\partial_t+nx\partial_x+u\partial_u\rangle$
\\
\hline
3
&
${\mathfrak g}^{\,}_1=\langle\partial_x\rangle,
\quad
{\mathfrak g}^\sigma_{3.1}=\langle\partial_t+\sigma\partial_x\rangle,
\quad
{\mathfrak g}^{\,}_{3.2}=\langle3nt\partial_t+nx\partial_x-2u\rangle$
\\
\hline
\end{tabular}
\end{center}
In all cases $a\in\mathbb R$, $n\ne0$,  $\sigma\in\{-1,0,1\}$.
\end{table}

To construct an ansatz associated with
a basis  Lie symmetry generator $\Gamma=\tau\partial_t+\xi\partial_x+\eta\partial_u$ of one-dimensional subalgebra from optimal list
we should find a solution of the so-called  invariant surface condition
$\Gamma[u]=\tau u_t+\xi u_x-\eta=0$. In practise, the corresponding characteristic system
$\frac{{\rm d}t}{\tau}=\frac{{\rm d}x}{\xi}=\frac{{\rm d}u}{\eta}$
should be solved.

We do not consider reductions associated with the subalgebra $\mathfrak g_1=\langle\partial_x\rangle$
because they lead to constant solutions only.
Ansatzes and reduced equations that are obtained for equations from class~\eqref{EqSenthilkumaran&Pandiaraja&Vaganan2008_h=0}
by means of one-dimensional subalgebras from Table~3 are collected in Table~4.

\begin{table}[h!]\small \renewcommand{\arraystretch}{1.65}
\begin{center}
\textbf{Table 4.} Similarity reductions of the class~$u_t+u^nu_x+g(t)u_{xxx}=0$ with $ng\neq0$
\\
that correspond to the subalgebras presented in Table 3.
\\[2ex]
\begin{tabular}
{|c|c|l|c|c|l|l|}
\hline
\hfil no.&\hfil$g(t)$ &\hfil ${\mathfrak g}$ & $\omega$ &\hfil Ansatz, $u$ &\hfil Reduced ODE
\\
\hline
1&$\varepsilon t^{\rho},\;\rho\ne-1$ & ${\mathfrak g}^{\,}_{1.1}$ & $xt^{-\frac{\rho+1}{3}}$ &
$t^{\frac{\rho-2}{3n}}\varphi(\omega)$ & $\varepsilon \varphi'''+\left(\varphi^n-\frac{\rho+1}{3}\omega\right)\varphi'+\frac{\rho-2}{3n}\varphi=0$
\\
\hline
2&$\varepsilon t^{-1}$ & ${\mathfrak g}^a_{1.2}$ & $x-\frac{a}{n}\ln t$ & $t^{-\frac{1}{n}}\varphi(\omega)$ & $\varepsilon
\varphi'''+\left(\varphi^n-\frac{a}{n}\right)\varphi'-\frac1nf=0$
\\
\hline
3&$\varepsilon e^t$ & ${\mathfrak g}^{\,}_{2}$ & $xe^{-\frac13t}$ & $e^{\frac1{3n}t}\varphi(\omega)$ & $\varepsilon
\varphi'''+\left(\varphi^n-\frac1{3}{\omega}\right)\varphi'+\frac1{3n}\varphi=0$
\\
\hline
4&$\varepsilon$ & ${\mathfrak g}^\sigma_{3.1}$ & $x-\sigma t$ & $\varphi(\omega)$ & $\varepsilon \varphi'''+(\varphi^n-\sigma)\varphi'=0$
\\
\hline
5&$\varepsilon$ & ${\mathfrak g}^{\,}_{3.2}$ & $xt^{-\frac13}$ & $t^{-\frac2{3n}}\varphi(\omega)$ &
$\varepsilon \varphi'''+\left(\varphi^n-\frac1{3}{\omega}\right)\varphi'-\frac2{3n}\varphi=0$
\\
\hline
\end{tabular}
\end{center}
In all cases $a\in\mathbb R$, $n\ne0$,  $\sigma\in\{-1,0,1\}$, $\varepsilon=\pm1\bmod G^\sim.$
\end{table}
It is possible to get exact traveling wave solutions of the equation $u_t+u^nu_x+\varepsilon u_{xxx}=0$ solving   the reduced ODE from Case~4 of Table~4,
\begin{equation}
\varepsilon \varphi'''+(\varphi^n-\sigma)\varphi'=0.
\end{equation}
If $\sigma\neq0$, two  partial  solutions of this equation are, for example,
$\varphi=\left(-\frac{\sigma(n+1)(n+2)}{2\sinh^2\left(\frac n2\sqrt{\frac {\sigma}{\varepsilon}}\omega+C\right)}\right)^\frac1n,$ and $\varphi=\left(\frac{\sigma(n+1)(n+2)}{2\cosh^2\left(\frac n2\sqrt{\frac {\sigma}{\varepsilon}}\omega+C\right)}\right)^\frac1n,$ where $C$ is an arbitrary constant.
They lead to the following traveling wave solutions
\[u=\left(-\frac{\sigma(n+1)(n+2)}{2\sinh^2\left(\frac n2\sqrt{\frac \sigma{\varepsilon}}(x-at)+C\right)}\right)^\frac1n,\quad u=\left(\frac{\sigma(n+1)(n+2)}{2\cosh^2\left(\frac n2\sqrt{\frac \sigma{\varepsilon}}(x-at)+C\right)}\right)^\frac1n,\]
of the generalized KdV equation with constant coefficients, $u_t+u^nu_x+\varepsilon u_{xxx}=0$. Note that some exact solutions were constructed in the literature for variable coefficient generalized KdV equations of the form $u_t+u^nu_x+\alpha(t) u+\varepsilon e^{-n\!\int\!\alpha(t)\,{\rm d}t} u_{xxx}=0$, (see, e.g.,~\cite{Biswas,Yang}), but it is due to the fact that the latter equations are reduced to constant coefficient ones.
Only a few exact solutions  for truly  variable coefficient generalized KdV equations are known.

\section{Boundary value problem for generalized KdV equations}
There are several approaches for exploiting Lie symmetries to reduce boundary value problems (BVPs) for PDEs to those for ODEs.
The classical technique
is to require that both equation and boundary conditions are left invariant under the one-parameter Lie group
of infinitesimal transformations. Of course, the infinitesimal approach is usually applied  (see, e.g.,~\cite[Section 4.4]{Bluman&Anco2002}).
We apply this technique for an IBVP for variable coefficient generalized KdV equations with and without linear damping terms.

\subsection{IBVP for generalized KdV without linear damping}
We apply at first this procedure to a problem with the governing equation being the KdV of the form~(\ref{EqSenthilkumaran&Pandiaraja&Vaganan2008_h=0}), i.e., we consider the following initial and boundary value problem
\begin{gather}\label{BV_KdV}
u_t+u^nu_x+g(t)u_{xxx}=0,\quad t>0,\quad x>0, \\[1ex]\label{BV_KdV_bc}
\begin{array}{@{}ll}
u(x,0)=0, &x>0,\\[0.5ex]
u(0,t)=q(t),& t>0,\\[0.5ex]
u_x(0,t)=0, &t>0,\\[0.5ex]
u_{xx}(0,t)=0, & t>0,
\end{array}
\end{gather}
where $q(t)$ is a nonvanishing smooth function of its variable.

The procedure starts by assuming a general symmetry of the form
\begin{equation}\label{general_symmetry}
\Gamma =\sum_{i=1}^n\alpha_i\Gamma_i,
\end{equation}
where $n$ is the number of basis operators of the maximal Lie invariance algebra of the given PDE and $\alpha_i,~i = 1,\dots,n$, are constants to be determined.

We consider Case 1 of Table 1. The general symmetry (\ref{general_symmetry}) takes the form
\[
\Gamma=\alpha_1\partial_x+\alpha_2\left [3nt\partial_t+(\rho+1)n x\partial_x+(\rho-2) u\partial_u \right ].
\]
Application of $\Gamma$ to the first boundary condition
$x=0,~u(t,0)=q(t)$ gives
\[
\alpha_1=0~~~~\mbox{and}~~~~q(t)=\gamma t^{\frac{\rho-2}{3n}},\quad \gamma=\const.
\]
Using the second extension of $\Gamma$,
\begin{gather*}
\Gamma^{(2)}=3nt\partial_t+(\rho+1)n x\partial_x+(\rho-2) u\partial_u\\\phantom{\Gamma^{(2)}=\,} +(\rho-n\rho-n-2)u_x\partial_{u_x}+(\rho-2n\rho-2n-2)u_{xx}\partial_{u_{xx}},
\end{gather*}
where the unused terms have been ignored, it can be shown that it leaves the initial condition and the remaining three boundary conditions invariant. Finally, (see Case 1 of Table~5) symmetry $\Gamma$ produces the ansatz
\begin{equation}
u=t^{\frac{\rho-2}{3n}}\varphi(\omega),\qquad\omega=xt^{-\frac{\rho+1}3},
\label{eq:Transf_BV_KdV}
\end{equation}
which reduces the problem (\ref{BV_KdV}) into the initial value problem
\begin{gather}\arraycolsep=0ex
\begin{array}{l}
\varepsilon \varphi'''+\varphi^n\varphi'-\frac{\rho+1}3\omega\varphi'+\frac{\rho-2}{3n}\varphi=0, \\[2ex]
\varphi(0)=\gamma, \qquad
 \varphi'(0)=0, \qquad
 \varphi''(0)=0.
 \end{array}
 \label{eq:IVP_from_BV_KdV}
\end{gather}


In Case 2 of Table~1, the corresponding symmetry does not leave the boundary conditions invariant and for Case 3 we obtain the above results with $\rho=0$.

\subsection{IBVP for generalized KdV with linear damping}

We consider the  IBVP for the generalized KdV equation with variable-coefficient linear damping
\begin{gather}\label{BV_KdV2}
u_t+u^nu_x+\frac j{t}u+g(t)u_{xxx}=0,\quad t>0,\quad x>0,\quad 
\end{gather}
with initial and boundary conditions~\eqref{BV_KdV_bc}. In the previous section we have shown that the symmetry operator which is admitted by both an equation from
class~\eqref{EqSenthilkumaran&Pandiaraja&Vaganan2008} and initial and boundary conditions~\eqref{BV_KdV_bc} with $q$ being a power function is the so-called dilatation operator, i.e., the operator corresponding to the one-parameter Lie group of scalings of the variables $t$, $x$ and $u$.
Equation~\eqref{BV_KdV2} admits a Lie symmetry generator which keeps the boundary conditions invariant if and only if $g$ is a power function or constant (Cases 1 and 3 of Table~2). Substituting $h=j/t$ into the formulas for $g$ and corresponding symmetry generators presented in Case 1 of Table~2 (without loss of generality we set $\kappa=0$) we find that
the equation
\begin{gather}\label{BV_KdV2}
u_t+u^nu_x+\frac j{t}u+\lambda t^{\rho(1-nj)-nj}u_{xxx}=0,
\end{gather}
admits the Lie symmetry generators $\partial_x$ and
\[\Gamma=\frac{3n}{1-nj}t\partial_t+n(\rho+1)x\partial_x+\left(\rho-2-\frac{3nj}{1-nj}\right)u\partial_u.\]
Boundary conditions~\eqref{BV_KdV_bc}  are left invariant with respect to the symmetry transformation generated by the operator $\Gamma$ if and only if  $q=\gamma t^{\frac{(\rho-2)(1-nj)-3nj}{3n}}$, where $\gamma=\const.$ Therefore we can apply Lie symmetries to solve the following BVP for the generalized KdV equation with linear damping
\begin{gather}\label{BV_KdV3}
u_t+u^nu_x+\frac j{t}u+\lambda t^{\rho(1-nj)-nj}u_{xxx}=0,\quad t>0,\quad x>0, \\[1ex]\label{BV_KdV_bc3}
\begin{array}{@{}ll}
u(x,0)=0, &x>0,\\[0.5ex]
u(0,t)=\gamma t^{\frac{\rho(1-nj)-nj-2}{3n}},& t>0,\\[0.5ex]
u_x(0,t)=0, &t>0,\\[0.5ex]
u_{xx}(0,t)=0, & t>0.
\end{array}
\end{gather}
The symmetry $\Gamma$ produces the transformation
\[
u=t^{\frac{\rho(1-nj)-nj-2}{3n}}\varphi(\omega),\qquad\omega=xt^{\frac{(\rho+1)(nj-1)}3},
\]
which reduces the problem (\ref{BV_KdV3})--(\ref{BV_KdV_bc3}) to
\begin{gather}\arraycolsep=0ex\label{IVP2}
\begin{array}{l}
\lambda \varphi'''+\varphi^n\varphi'+\frac{(\rho+1)(nj-1)}3\omega\varphi'+\frac{(\rho-2)(1-nj)}{3n}\varphi=0, \\[2ex]
\varphi(0)=\gamma, \qquad
 \varphi'(0)=0, \qquad
 \varphi''(0)=0.
 \end{array}
\end{gather}

For $j=1/2$ and $j=1$ (\ref{BV_KdV3}) becomes generalized cylindrical and spherical KdV equations, respectively.

\section{Numerical solution}
In the previous section IBVPs for certain generalized KdV equations were reduced to initial value problems for third-order nonlinear ODEs.
Now we can use well-developed numerical techniques (e.g. finite differences) to solve the reduced problem.

\subsection{Finite-difference scheme}
We focus our numerical investigation on solving the initial value problem (\ref{eq:IVP_from_BV_KdV}).
\noprint{ \begin{gather}\arraycolsep=0ex
\begin{array}{l}
\varepsilon \varphi'''+\varphi^n\varphi'-\frac{\rho+1}3\omega\varphi'+\frac{\rho-2}{3n}\varphi=0, \\[2ex]
\varphi(0)=\gamma, \qquad
 \varphi'(0)=0, \qquad
 \varphi''(0)=0.
 \end{array}
\end{gather}}
Then, with the aid of transformation (\ref{eq:Transf_BV_KdV}) it is easy to evaluate the solution $u(x,t)$ of problem (\ref{BV_KdV})--(\ref{BV_KdV_bc}). The consideration for~\eqref{IVP2}
can be performed in an analogous way.

Utilizing second-order finite difference approximations for the first and third order derivative, namely,
\begin{gather*}
\varphi'=\frac{\varphi_{i+1}-\varphi_{i-1}}{2h}-\frac{h^2}{6}\varphi'''|_{\omega_i}+\mathcal{O}(h^3),\\
\varphi'''=\frac{\varphi_{i+2}-2\varphi_{i+1}+2\varphi_{i-1}-\varphi_{i-2}}{2h^3}
-\frac{h^2}{4}\varphi'''''|_{\omega_i}+\mathcal{O}(h^3),
\end{gather*}
we arrive at the following numerical scheme:
\begin{equation*}
   \begin{split}
   \frac{\varepsilon}{2 h^3}(\varphi_{i+2} &- 2\varphi_{i+1} + 2\varphi_{i-1} - \varphi_{i-2})\\
   &+ \frac{\varphi_i^n}{2h}(\varphi_{i+1}-\varphi_{i-1}) - \frac{\rho+1}{6h}\omega_i(\varphi_{i+1}-\varphi_{i-1})
   + \frac{\rho-2}{3n}\varphi_i = 0.
   \end{split}
\end{equation*}  
We then assume a general interval of interest for $\omega$, say $\omega \in [a,b]$, and
define the grid for $\omega$
\begin{equation}
   \omega_i = a + (i-1)h\ , \quad h=\frac{b-a}{N}\ , \quad i=1,2,3,\ldots,N,N+1.
   \label{eq:grid_1}
\end{equation}
The initial conditions are discretized as follows:
\begin{subequations}
\label{eq:SE_FDScheme}
\begin{equation}
   \varphi_1 = \gamma, \qquad \varphi_1 = \varphi_2, \qquad \varphi_3 = \varphi_2.
   \label{eq:SE_FDScheme_ICs}
\end{equation}
Clearing the denominators and solving for $\varphi_{i+2}$ yields the following simple explicit scheme

\begin{equation}
      \begin{split}
      \varphi_{i+2} &= \frac{ 1}{ \varepsilon}\left\{\left(2 \varepsilon + \frac{\rho+1}{3}h^2\omega_i\right)\varphi_{i+1} - \frac{2 h^3}{3n}(\rho-2)\varphi_i  -  \left(2 \varepsilon + \frac{\rho+1}{3}h^2\omega_i\right)\varphi_{i-1} \right.\\
      &\qquad\quad+ \left. \varepsilon\varphi_{i-2} + g_i  \right\},
      \raisetag{1\baselineskip}
      \label{eq:SE_FDScheme_main}
      \end{split}
   \end{equation}
where $g_i = -h^2\varphi_i^n(\varphi_{i+1} - \varphi_{i-1})$ and $i=3,4,5,\ldots,N-1$. To initialize the iteration process for the spatial dimension $\omega$,  we set $i=2$ in (\ref{eq:SE_FDScheme_main}) and obtain
\begin{equation}
      \begin{split}
      \varphi_{4} &= \frac{1}{ \varepsilon}\left\{\left(2 \varepsilon + \frac{\rho+1}{3}h^2\omega_2\right)\varphi_{3} - \frac{2 h^3}{3n}(\rho-2)\varphi_2  -  \left(2 \varepsilon + \frac{\rho+1}{3}h^2\omega_2\right)\varphi_{1} \right.\\
      &\qquad\quad+ \left. \varepsilon\varphi_{0} + g_2 \right\},
      \raisetag{1\baselineskip}
      \label{eq:SE_FDScheme_init}
      \end{split}
   \end{equation}
\end{subequations}
where $g_2 = -h^2\varphi_2^n(\varphi_{3} - \varphi_{1})$ and $\varphi_0 = \varphi_1$. The value of $\varphi_0 = \varphi(-h)$ was computed by applying central and/or backward finite differences to the initial conditions.

The nonlinear system of difference equations (\ref{eq:SE_FDScheme}) is solved using fixed-point iterations
\begin{subequations}
    \label{eq:SE_FDScheme__Iter}
   \begin{equation}
      \begin{split}
      \varphi_{i+2}^{k+1} &= \frac{1}{\varepsilon}\left\{\left(2 \varepsilon + \frac{\rho+1}{3}h^2\omega_i\right)\varphi_{i+1}^{k+1} - \frac{2 h^3}{3n}(\rho-2)\varphi_i^{k+1}  -  \left(2 \varepsilon + \frac{\rho+1}{3}h^2\omega_i\right)\varphi_{i-1}^{k+1} \right.\\
      &\qquad\quad+ \left. \varepsilon\varphi_{i-2}^{k+1} + g_i^k  \right\},
      \raisetag{1\baselineskip}
      \label{eq:SE_FDScheme_main_Iter}
      \end{split}
   \end{equation}
   \begin{equation}
      \begin{split}
      \varphi_{4}^{k+1} &= \frac{1}{\varepsilon}\left\{\left(2 \varepsilon + \frac{\rho+1}{3}h^2\omega_2\right)\varphi_{3}^{k+1} - \frac{2 h^3}{3n}(\rho-2)\varphi_2^{k+1}  -  \left(2 \varepsilon + \frac{\rho+1}{3}h^2\omega_2\right)\varphi_{1}^{k+1} \right.\\
      &\qquad\quad+ \left. \varepsilon\varphi_{0}^{k+1} + g_2^k \right\}.
      \raisetag{1\baselineskip}
      \label{eq:SE_FDScheme_init_Iter}
      \end{split}
   \end{equation}
\end{subequations}
The algorithm is initialized at $k=0$ using the initial condition, i.e., $\varphi^0_1=\varphi^0_2=\varphi^0_3=\gamma$ and $\varphi^0_i=0$ for $i=4,\ldots,N,N+1$. Then, for $k=1,2,3,4,\ldots$, the values of $\varphi_i^k$ are updated for all $i=4,\ldots,N,N+1$ except $\varphi_j^k\,, \ j=1,2,3$, which remain fixed at the initial value $\gamma$ for all $k$.  The algorithm is terminated  when the following convergence criterion is satisfied
   \[ \frac{\max\limits_{1\leq i \leq N}|\varphi_i^{k+1}-\varphi_i^{k}|}{\max\limits_{1\leq i \leq N}|\varphi_i^{k+1}|}\leq 10^{-8}.\]

An exact analytic solution for IVP (\ref{eq:IVP_from_BV_KdV}) is not known for any admissible set of parameters except for when $\rho=2$, which yields the constant solution. The accuracy of numerical scheme (\ref{eq:SE_FDScheme}) is assessed by applying the scheme for different mesh densities (see Fig.~\ref{fig:ErrorVsNMAX_SE}). We have chosen the parameter set $n=1$, $\rho=1$, $\varepsilon=-1$ for this investigation. The initial condition value was set at $\gamma=0.5$ and the solution interval was specified as $[0,50]$. We will refer to the solution of (\ref{eq:IVP_from_BV_KdV}) for these parameter values as the ``standard solution".
  \begin{figure}[t!]
  \vspace*{-3mm}
  \centerline{\hspace{-4mm}
  \includegraphics[width=0.7\textwidth]{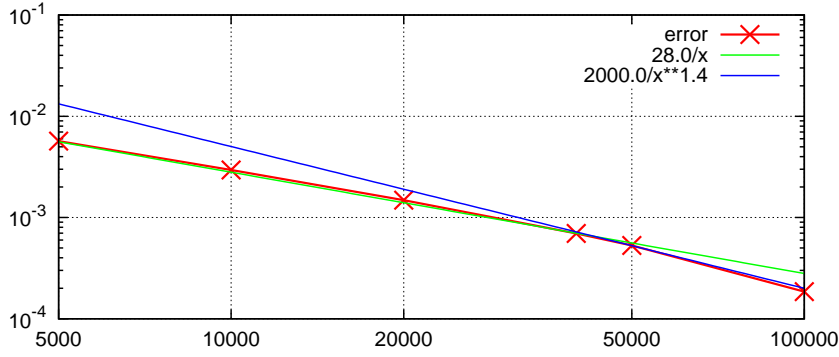}}
  \vspace*{-3mm}
   \caption{Logarithmic plot of the absolute error as a function of the total number of gridpoints~$N$. Governing parameters: $n=1$, $\rho=1$, $\varepsilon=-1$, $\gamma=0.5$ and $[a,b]=[0,50]$ (standard solution). }
   \label{fig:ErrorVsNMAX_SE}
  \end{figure}
The part of the exact solution is played by the numerical solution obtained on a very fine grid $\hat{\varphi}$. Here, $\hat{\varphi}$ is the solution for $N=200000$ where $N$ is the number of grid points used to discretize our interval (see (\ref{eq:grid_1})). The absolute error for a numerical solution $\varphi^N$ obtained using $N$ grid points is defined by $\max\limits_{1\leq i \leq N}\|\hat{\varphi_i} - \varphi_{i,N}\|$. It is observed that the absolute error is inversely proportional to $N$. This is slightly improved for $N \geq 40000$.

Our numerical method was further validated by applying it to a similar nonlinear initial value problem of the same order that appeared in the literature \cite{AbdelMalek&Helal2011} in the context of solving IBVP for generalized Burgers--KdV equations. This reads
\begin{subequations}
   \label{eq:Mina_IVP}
   \begin{equation}
       \beta F''' - F'' - \alpha n F^{n-1} F' + \tfrac{1}{2}\eta F' + \tfrac{1}{2n-2} F = 0  ,
      \label{eq:Mina_Cubic_ODE}
   \end{equation}
   \begin{equation}
       F(0) = \gamma ,\quad F'(0) = 0 ,\quad F''(0) = 0 .
       \label{eq:Mina_ICs}
   \end{equation}
\end{subequations}
where $F=F(\eta)$ is the sought function, $\eta$ the independent variable and $\beta$, $\alpha$, $n$ and $\gamma$ parameters. In Fig.~\ref{fig:Mina_Comparison} we present a comparison between our results and \cite{AbdelMalek&Helal2011} for the case $n=2$, $\alpha=1$, $\beta=10$, $\gamma=0.5$. It is easy to see that the results are identical. This was observed for all parameter values investigated.
 \begin{figure}[h!]
  \vspace*{-4mm}
  \centerline{\hspace{-4mm}
  \includegraphics[width=0.7\textwidth]{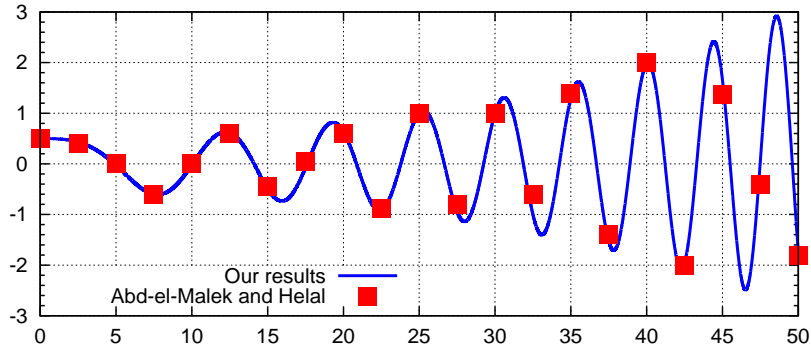}}
  \vspace*{-3mm}
   \caption{Solution of (\ref{eq:Mina_IVP}) for $n=2$, $\alpha=1$, $\beta=10$, $\gamma=0.5$ and $[a,b]=[0,50]$.  Squares: The solution in \cite{AbdelMalek&Helal2011} sampled at equidistant nodes of distance 2.5 on the $\eta$ axis. Solid line: Our results.}
   \label{fig:Mina_Comparison}
  \end{figure}

\subsection{Numerical results}

We investigate the effect of varying the governing parameters in (\ref{eq:IVP_from_BV_KdV}). Because the parametric space is vast, the method we follow here is to vary only one parameter whilst keeping all others fixed at the so-called ``standard solution'' values ($n=1$, $\rho=1$, $\varepsilon=-1$, $\gamma=0.5$, $[a,b]=[0,50]$). All solutions in this subsection were computed for $N=100000$ $\omega$-grid points.

  \begin{figure}[h!]
    \vspace*{-4mm}
    \centerline{\hspace{-4mm}
    \includegraphics[width=0.7\textwidth]{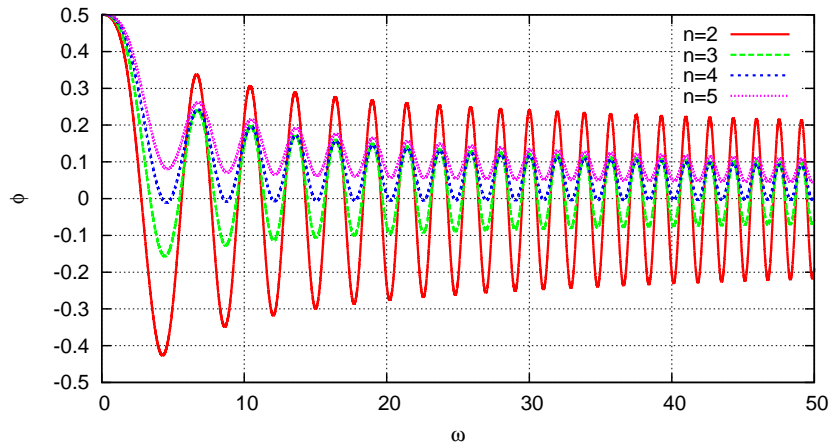}}
    \vspace*{-3mm}
    \caption{Solution of (\ref{eq:IVP_from_BV_KdV}) for various values of the power $n$ of the nonlinear term.}
    \label{fig:various_nonlin_power}
  \end{figure}
  \begin{figure}[h]
    \vspace*{-4mm}
    \centerline{\hspace{-4mm}
    \includegraphics[width=0.7\textwidth]{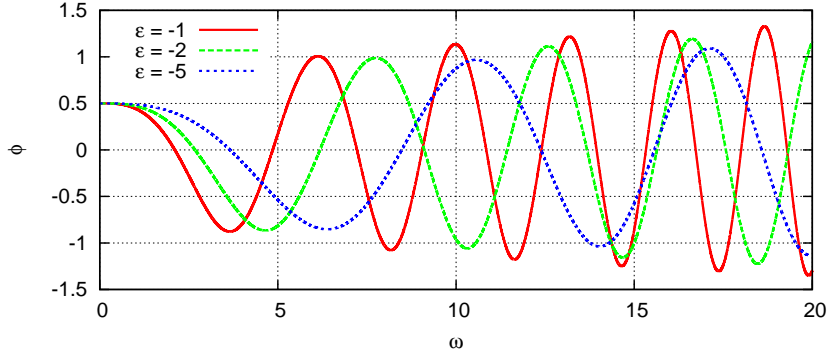}}
    \vspace*{-3mm}
    \caption{Solution of (\ref{eq:IVP_from_BV_KdV}) for $\varepsilon = -1, -2, -5$.}
    \label{fig:various_epsilon}
  \end{figure}

In Fig.~\ref{fig:various_nonlin_power} we present the effect of the order $n$ of the nonlinear term. Our solution can be described as a damped oscillation. Increasing $n$ results in a decrease of the solution ``amplitude'', however, the ``wavelength'' is unaffected.

The value of the coefficient $\varepsilon$ of the third derivative  affects the solution in two ways (see Fig.~\ref{fig:various_epsilon}):  Increasing negative values of $\varepsilon$ in magnitude, widens the oscillations, i.e., increases the solution ``wavelength''. In addition, the oscillation ``amplitude'' decreases. It was observed that the solutions for positive epsilon values are oscillatory but with an increasing amplitude (divergent). This is an innate feature of the problem and not a numerical artifact and was also observed in \cite{AbdelMalek&Helal2011} for similar cases.

\vspace{2mm}

 \begin{figure}[h]
    \vspace*{-4mm}
    \centerline{\hspace{-4mm}
    \includegraphics[width=0.7\textwidth]{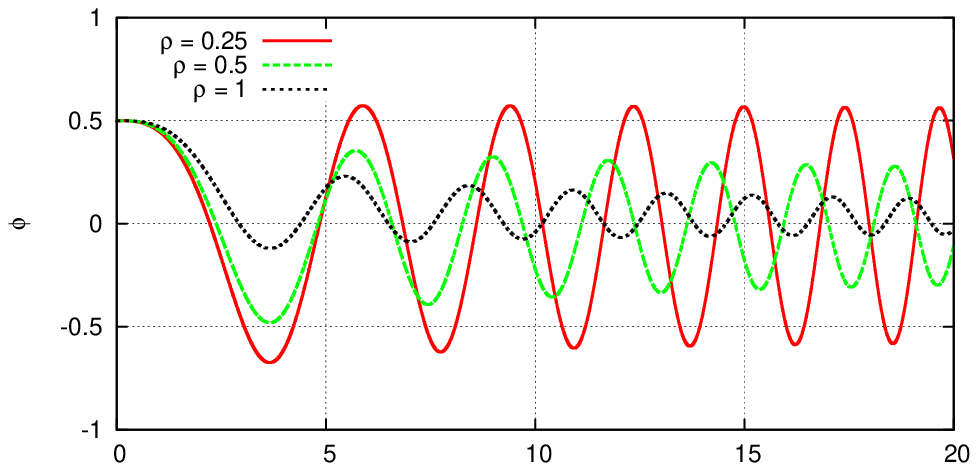}}
     \centerline{\hspace{-3mm}
     \includegraphics[width=0.7\textwidth]{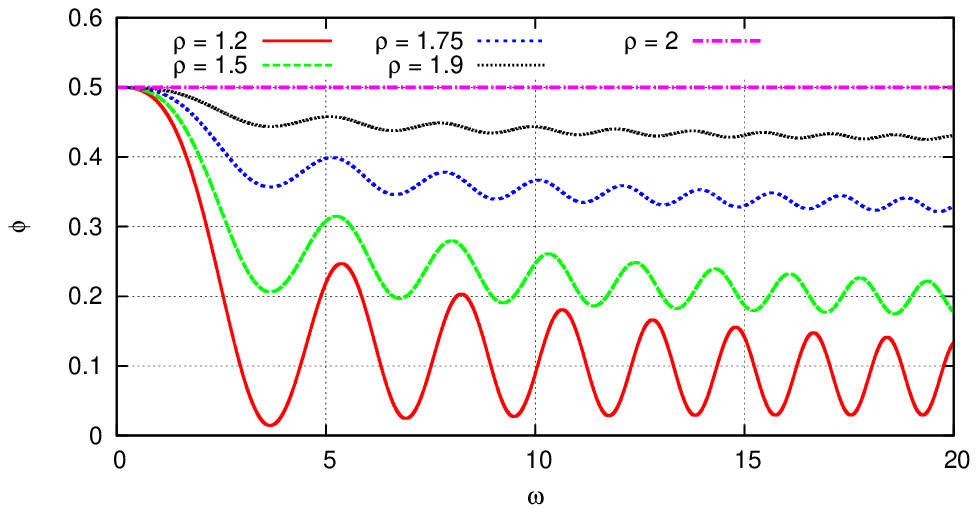}}
    \vspace*{-4mm}
    \caption{Top panel: Solution of (\ref{eq:IVP_from_BV_KdV}) for $\rho = 0.5, 0.25, 1$. Bottom panel: Solution of (\ref{eq:IVP_from_BV_KdV}) for $\rho = 1.2, 1.5, 1.75, 1.9, 2$.}
    \label{fig:various_rho_1}
 \end{figure}

Varying $\rho$ has a drastic effect on the solution form (see Fig. 5). We have found that for $\rho < 0$ the solution diverges abruptly to infinity and is not oscillatory. For $\rho > 2$  the solutions are divergent but oscillatory resembling the results published in \cite{AbdelMalek&Helal2011}. As expected, for $\rho=2$ we obtain the constant solution $\varphi = 2$, which can be easily verified analytically. Increasing $\rho$ from 0 to 1 decreases the oscillation amplitude. Specifically, for $\rho \in [0.5,2)$ we get decreasing oscillations and for $\rho \in [1,2)$ the decreasing oscillations are so pronounced the solution resembles a soliton with a tail.

Having computed the numerical solution of (\ref{eq:IVP_from_BV_KdV}) we construct the solution $u(x,t)$ of (\ref{BV_KdV})--(\ref{BV_KdV_bc}) utilizing (\ref{eq:Transf_BV_KdV}). The case illustrated in Figs.~\ref{fig:3D_Surface} and~\ref{fig:Fence_Plot} corresponds to the so-called ``standard solution''. We observe that the function $u(x,\cdot)$ for fixed $t$ is an oscillatory function of $x$ whereas for fixed $x$, $u(\cdot,t)$ decreases asymptotically with $t$.
\begin{figure}[h!]
    \vspace*{-4mm}
    \centerline{\hspace{-4mm}
    \includegraphics[width=0.7\textwidth]{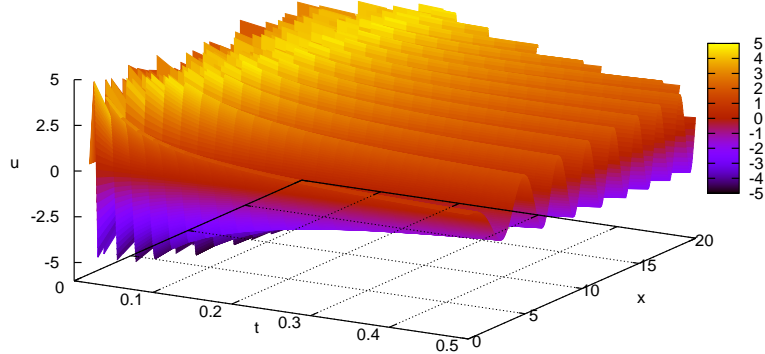}}
    \vspace*{-4mm}
    \caption{Surface plot of $u=u(x,t)$ for $n=1$, $\rho=1$, $\varepsilon=-1$, $\gamma=0.5$. The range was truncated to $-5\leq u \leq 5$. }
    \label{fig:3D_Surface}
\end{figure}
\vspace*{-5mm}
\begin{figure}[h!]
    \vspace*{-4mm}
    \centerline{\hspace{-4mm}
    \includegraphics[width=0.7\textwidth]{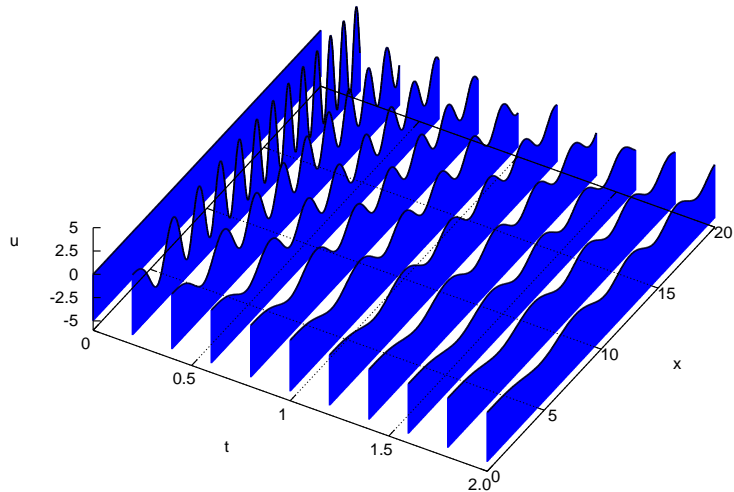}}
    \vspace*{-5mm}
    \caption{Fence plot of the case presented in Fig.~\ref{fig:3D_Surface}. The snapshots of our solution were taken at $t_i = i \Delta t$, where $i=0,1,2,\ldots,10$ and $\Delta t = 0.2$.}
    \label{fig:Fence_Plot}
\end{figure}


\vspace{-5mm}

\section{Conclusion}
Lie symmetry methods are a well-known powerful tool for the construction of exact solutions for nonlinear differential equations, both ordinary and partial. Unfortunately, it is often not known how to solve the reduced differential equation analytically. This is the case for variable coefficient generalized KdV equations that are inequivalent by point transformations  to constant  coefficient ones.
 Nevertheless, Lie symmetry methods can be useful even in those cases when reduced ODEs cannot be solved analytically but only numerically.
We illustrate this fact considering initial and boundary value problem for subclasses of generalized KdV equations admitting Lie symmetries extension.

A group classification for variable coefficient generalized KdV equations~\eqref{EqSenthilkumaran&Pandiaraja&Vaganan2008} is carried out exhaustively. The results are presented in two ways: up to $G^\sim$-equivalence (Table~1) and without simplification by equivalence transformations (Table~2). Similarity solutions are classified (Table~4).
The derived Lie symmetries of a generalized KdV equation are employed  to IBVP  (\ref{BV_KdV})--(\ref{BV_KdV_bc})   transforming
it into an IVP for an ODE~\eqref{eq:IVP_from_BV_KdV}. The resulting nonlinear problem is solved numerically with the aid of a second-order
finite-difference scheme with fixed-point iterations. The scheme was validated by applying it to similar problems
in the literature which were solved using other methods and the results were found to be in excellent agreement.
The effect of the governing parameters on the solutions of (\ref{BV_KdV})--(\ref{BV_KdV_bc}) was examined and solutions of the original PDE were constructed using the aforementioned transformations.

\subsection*{Acknowledgements}

OOV  expresses the gratitude to the hospitality
shown by the University of Cyprus during her visits to the University. The authors also wish to thank Profs. R. Popovych and V. Boyko for useful comments.



\begin{thebibliography}{99}\itemsep=1ex
\footnotesize
\bibitem{AbdelMalek1998}
M.B. Abd-el-Malek, Application of the group-theoretical method to physical problems, {\it J. Nonlinear Math. Phys.}  {\bf 5}  (1998), 314--330.

\bibitem{AbdelMalek&Helal2011} M.B. Abd-el-Malek and M.M. Helal, Group method solutions of the generalized forms of Burgers, Burgers-KdV and KdV equations with time-dependent variable coefficients, {\it Acta Mechanica} {\bf 221} (2011), 281--296.

\bibitem{Biswas} A. Biswas, Solitary wave solution for the generalized KdV equation with time-dependent damping and dispersion, {\it Commun. Nonlinear Sci. Numer. Simulat.} {\bf 14} (2009), 3503--3506.

\bibitem{Bluman&Cole1969} G.W. Bluman and J.D. Cole, The general similarity solution of the heat equation, {\it J. Math. Mech.} {\bf 18} (1969), 1025--1042.

\bibitem{Bluman&Anco2002} G.W. Bluman and S.C. Anco, {\it Symmetry and integration methods for differential equations}, Springer-Verlag, New York, 2002.

\bibitem{Bluman1974} G.W. Bluman, Application of the general similarity solution of the heat equation to boundary-value problems, {\it Quart. Appl. Math.} {\bf 31} (1974), 403--415.


\bibitem{Gungor&Lahno&Zhdanov2004}
 F. G\"ung\"or, V.I. Lahno and R.Z. Zhdanov,
Symmetry classification of KdV-type nonlinear evolution equations,
{\it J. Math. Phys.}    {\bf 45} (2004) 2280--2313;  arXiv:nlin/0201063.


\bibitem{Kingston&Sophocleous1998}
J.G. Kingston and C. Sophocleous, On form-preserving point
transformations of partial differential equations, {\it J. Phys. A:
Math. Gen.} {\bf 31} (1998), 1597--1619.


\bibitem{Magadeev1993}
 B.A. Magadeev, On group classification of nonlinear evolution equations,
{\it Algebra i Analiz} {\bf 5} (1993) 141--156 (in Russian);
translation in {\it St. Petersburg Math. J.}    {\bf 5} (1994) 345--359.







\bibitem{Olver1986}
P. Olver, {\it Applications of Lie groups to differential equations},
New-York, Springer-Verlag, 1986.


\bibitem{Ovsiannikov1982}
L.V. Ovsiannikov, {\it Group analysis of differential equations}, New
York, Academic Press, 1982.


\bibitem{pate1977a}
J. Patera, P. Winternitz,
Subalgebras of real three- and four-dimensional Lie algebras,
{\it J.~Math. Phys.} {\bf 18} (1977), no.~7, 1449--1455.


\bibitem{Popovych&Ivanova2004}
R.O. Popovych and N.M. Ivanova, New results on group classification of nonlinear diffusion-convection equations, {\it J. Phys. A} {\bf 37} (2004), 7547--7565; arXiv:math-ph/0306035.

\bibitem{popo2010a}
R.O. Popovych, M. Kunzinger and H. Eshraghi,
Admissible transformations and normalized classes of nonlinear Schr\"odinger equations,
{\it Acta Appl. Math.} {\bf  109} (2010), 315--359, arXiv:math-ph/0611061.

\bibitem{Popovych&Vaneeva2010}
R.O. Popovych and O.O. Vaneeva,   More common errors in finding exact solutions of nonlinear differential equations: Part I, {\it Commun. Nonlinear Sci. Numer. Simulat.} {\bf 15} (2010), 3887--3899; arXiv:0911.1848.


\bibitem{Senthilkumaran&Pandiaraja&Vaganan2008}
 M. Senthilkumaran,  D. Pandiaraja and  B.M.  Vaganan,
New exact explicit solutions of the generalized KdV equations,
{\it Appl. Math. Comput.} {\bf 202} (2008), 693--699.





\bibitem{Vaneeva2012}
O.O. Vaneeva,  Lie symmetries and exact solutions of variable coefficient mKdV equations: an equivalence based approach, {\it Commun. Nonlinear Sci. Numer. Simulat.} {\bf 17} (2012), 611--618; arXiv:1104.1981.

\bibitem{Vaneeva2013}
O.O. Vaneeva,  Group classiffication of variable coeffficient
KdV-like equations, 451--459, V.~Dobrev (ed.), {\it Springer Proceedings in Mathematics {\rm \&} Statistics, Vol. 36. IX International Workshop ``Lie Theory and Its Application in Physics''}, Springer, 2013; arXiv:1204.4875.


\bibitem{vane2012b}
O.O. Vaneeva, R.O. Popovych and C. Sophocleous,
Extended group analysis of variable coefficient reaction-diffusion
equations with exponential nonlinearities,
{\it J. Math. Anal. Appl.} {\bf 396} (2012), 225--242; arXiv:1111.5198.

\bibitem{VPS2013} O.O. Vaneeva, R.O. Popovych and C. Sophocleous, Equivalence transformations in the study of integrability, arXiv:1308.5126




\bibitem{VSL}
O.O. Vaneeva, C. Sophocleous and P.G.L. Leach,
Lie symmetries of generalized Burgers equations: application to boundary-value problems,
arXiv:1303.3548.


\bibitem{Yang}
Y. Yang, Z.-L. Tao and F.R. Austin, Solutions of the generalized KdV equation with time-dependent
damping and dispersion,
{\it Appl. Math. Comput.} {\bf 216} (2010), 1029--1035.

\end{thebibliography}
\end{document}